\documentclass{amsart}
 
 \usepackage{graphicx}
\usepackage{color}
\usepackage{float}
\usepackage{subfig}
\usepackage{amssymb}
 \usepackage{amsmath}
 \usepackage{amscd}
 \usepackage{amsfonts}
 \usepackage{mathrsfs}
 \usepackage{amsthm}
 \usepackage{hyperref}
\usepackage{fullpage}

\theoremstyle{plain}
 \newtheorem{thm}{Theorem}[section]
 \newtheorem{lem}[thm]{Lemma}
 \newtheorem{prop}[thm]{Proposition}
 \newtheorem{cor}[thm]{Corollary}
\theoremstyle{definition}
 \newtheorem{definition}[thm]{Definition}
 \theoremstyle{remark}
  \newtheorem{rmk}[thm]{Remark}
 \newtheorem{exm}[thm]{Example}
 
    \newcommand{\sL}{\mathscr{L}}
        \newcommand{\cb}{\ulcorner}
        \newcommand{\rb}{\urcorner}
  \newcommand{\wt}{\widetilde}
 
  \newcommand{\Aut}{{\rm Aut}}
  \newcommand{\nms}{\negmedspace}

\title{The (Nested) Word Problem}
    
\author{Christopher S. Henry}

\begin{document}


\setcounter{secnumdepth}{1}
\setcounter{tocdepth}{2}

\begin{abstract} In this article we provide a new perspective on the word problem of a group by using languages of nested words. These were introduced by Alur and Madhusudan as a way to model programming languages such as HTML. We demonstrate how a class of nested word languages called visibly pushdown can be used to study the word problem of virtually free groups in a natural way. 

\noindent{\bf AMS MSC classes:} 20F10,20E05,68Q45,03D40.

\end{abstract}

\maketitle 

\vspace{-1cm}
\section{Introduction}\label{sec-intro}

There are deep connections between formal language theory and group theory, and this relationship has been used  to successfully explore the structure of groups by mathematicians. For example, in the foundational work of Epstein et al \cite{E+}, the authors  introduce the notion of automatic structures on groups which are defined in terms of regular languages. They show that many naturally occurring groups admit automatic structures, such as the braid groups, mapping class groups, and many 3-manifold groups. Moreover, they show that the word problem is solvable for any automatic group, and in fact  can be solved in quadratic time. Further, automatic groups share a number of other appealing algorithmic properties, and this is an example of a fruitful application of  formal  languages to the study of groups.

In order to describe the relationship between formal languages and groups, one typically uses a finite generating set for the group. For instance, given a class $\sL$ of formal languages, one interesting question is whether a given (finitely generated) group has word problem in $\sL.$ In \cite{MS} Muller and Schupp show that a finitely generated group $G$ has a word problem that is a context-free language if and only if $G$ is virtually free. Their work has been extended by considering other formal language classes, see \cite{He}, \cite{HT}, \cite{EO}, and \cite{HOT}, and similar results have been obtained by considering the complement of the word problem, see \cite{H+} and \cite{HR}.

In this article we extend the definition of the word problem to languages of nested words. These were first introduced by Alur and Madhusudan in \cite{AM} as a way to model programming languages such as HTML more efficiently. A nested word contains additional structure called a matching relation which specifies how letters in the word should be paired together. This additional structure means that the standard definition of word problem is not meaningful for nested words, so we provide an appropriate one. Using this definition we study groups with a word problem that is a visibly pushdown language of nested words (VPL). VPLs have a close relationship to both regular languages and context-free languages. Our main result is: \\

{\bf Theorem \ref{thm-nwp}}
A finitely generated group admits a visibly pushdown nested word problem if and only if it is virtually free. \\

In other words, groups with a VPL nested word problem are the same as those with a context-free word problem. However, we show that it can be more natural to consider nested words. For example in the free group any trivial word has a canonical matching relation associated to it, corresponding to successive cancellations of generators. Furthermore, VPLs have nicer closure properties than context-free languages. By proving some additional closure properties we demonstrate that trivial words for two different classes of virtually free groups also have a canonical matching relation. This  in turn provides a nice correspondence between the group theoretic properties of direct/semi-direct product and closure properties of  VPLs. 

The paper is organized as follows. In the second section, we provide the necessary background on formal languages. Standard material on regular and context-free languages is presented using pushdown automata, which allows for an easy introduction to nested words and VPLs. Section Three describes how formal languages are typically used to study the word problem.  In the fourth and final section, we introduce the nested  word problem and use VPLs to study the structure of virtually free groups.

\section{Preliminaries}

In this section, we review basic results on context-free languages, regular languages, and visibly pushdown languages. There are many different ways to describe context-free and regular languages, for example via grammars or regular expressions. Our presentation uses a machine theoretic approach which is convenient for introducing visibly pushdown languages. 
 \subsection{Formal Languages}\label{sec-fl}
  Let \( A \) be a finite set, which we will call an alphabet. For each \( n  \in \mathbb{N} \), we let \(A^n = \{ w \mid w: \{1,2,\dots,n\}   
 \rightarrow A \) is a function\}. An element \( w  \in A^n \) is called a word of length \(|w|=n\) and denoted by \( w = a_1 \cdots a_n \), where \( w(i) = a_i  \in A  \).
For \(1 \leq i < n\) let \(w[i]=a_1\cdots a_i\) be the prefix of \(w\) of length \(i\), and set \(w[i]=w\) for \(i \geq n\).
 Denote by \( \epsilon \)  the unique element \( \epsilon : \varnothing \rightarrow A \) of \( A^0 \) called the empty word. Finally let \( A^* = \bigcup_{n=0}^{\infty} A^n \) be the set of all finite words over the alphabet \( A \).
 
 \begin{definition}
 Given an alphabet \( A \), a language over \( A \) is any subset \( L \subset A^* \) .
 \end{definition}
 
 For \( L \) to be a meaningful collection of words one would expect (at least) an algorithm
 to recognize when an arbitrary word \(w \in A^*\) lies in \(L\). In fact, one way to define  formal languages classes is precisely 
 by the type of algorithm which recognizes words. The notion of algorithm is formalized by defining machines (automata), which can be thought of as reading in words and deciding whether or not they belong to the language.
 
 
 \subsection{Context-free Languages}
\begin{definition}\label{ch-def-1.10}
A  deterministic pushdown automaton (PDA) M is a tuple \\ 
\((A,S, \Gamma,s_0, \gamma_0,Y, \delta)\) such that:

 \begin{enumerate}

 \item  \( A \) is an alphabet,
 \item \( S \) is a finite set of states,
 \item \( \Gamma \) is a finite stack alphabet,
 \item \( s_0 \in S \) is the start state, and \(\gamma_0 \in \Gamma\) is the bottom of stack symbol,
 \item  \( Y \subset S \) is the set of accept states,
 \item \( \delta : D  \rightarrow S \times \Gamma^* \) is the transition function defined on a subset \(D \subset S \times (A \cup \{\epsilon\}) \times \Gamma\). \\ It is deterministic with respect to epsilon transitions, i.e.  \((s,\epsilon,\gamma) \in D \Rightarrow (s,a,\gamma) \notin D \) for all \(a \in A\).
 
\end{enumerate}
 \end{definition}
   
  To make transitions, a PDA needs to know the current state \(s\), the input symbol being read in \(a\), and the current top of stack symbol \(\gamma\). It transitions to a new state \(s'\), erases \(\gamma\), and replaces it on the top of the stack by a finite word \(\chi=\gamma_1 \cdots \gamma_n \in \Gamma^*\). We think of \(\chi\) as being added to the stack one letter at a time, starting with \(\gamma_1\) and ending with \(\gamma_n\) which becomes the new top of stack symbol. We interpret \(\delta(s,\epsilon,\gamma)\) as the machine performing a stack operation without having to read an input symbol. 

 \begin{definition}\label{ID-of-PDA}
Given a deterministic PDA \(M = (A,S, \Gamma,s_0, \gamma_0,Y, \delta)\), an instantaneous description of \(M\) is a triple \((s,w,\chi) \in S \times A^* \times \Gamma^*\).  
For \(a \in (A \cup \{\epsilon\})\) we let \((s,aw,\chi\gamma) \vdash (t,w,\chi\chi') \) if \(\delta(s,a,\gamma) = (t,\chi')\). Denote by \(\vdash^*\) the reflexive and transitive closure of \(\vdash\).  The language of words accepted by M is \[L(M) = \{w \in A^* \mid (s_0,w,\gamma_0) \vdash^*(y,\epsilon,\chi) \text{ for some } y \in Y \text{ and } \chi \in \Gamma^*\}.\] 
 \end{definition}

 \begin{definition}
 A language \(L \subset A^*\) is deterministic context-free (CF) if there exists some PDA \(M\) such that \(L=L(M)\). We denote the class of all deterministic context-free languages by \(\sL_{CF}\).

 \end{definition}
 
 \begin{exm}\label{ch-ex-1.14}
 
 Consider the alphabet \(A = \{a,b\}\). The language of words given by  \( L = \{ a^nb^n \mid n = 0, 1, \dots \} \) is CF, where \( a^n = \overbrace{a\cdots a}^\text{n times} \) and \( a^0=b^0=\epsilon \). A PDA accepting $L$ has stack alphabet
 \(\Gamma = \{0,1\}\) with \(\gamma_0 = 0\), and set of states  \(S = \{s_0,s_1,s_2,s_y,s_f\}\) with \(Y = \{s_0,s_y\}\). The transition function satisfies:
 
 \[
  \delta(s,x,\gamma) = \begin{cases}
         (s_i,a,i) \mapsto (s_1,i1)  & \text{ for $i\in \{0,1\}$},   \\
         (s_i,b,1) \mapsto (s_2,\epsilon) &  \text{ for $i\in \{1,2\}$}, \\
        (s_2,\epsilon,0)  \mapsto (s_y,\epsilon), & \\
        (s,x,\gamma) \mapsto (s_f,\gamma) & \text{in all other cases}.  
         \end{cases}
 \]
 \end{exm} 
  
The class of CF languages satisfies two important closure properties.
 
  \begin{prop}[Hopcroft and Ullman \cite{HU}]
  Let \(L \in \mathscr{L}_{CF}\) be a CF language over the alphabet
   \(A\). The following languages are also CF:
  \begin{enumerate}

  \item \(L_{\rm comp} = A^* \smallsetminus L \),
  \item \(L_{\rm pre} = \{w \mid wu = v \in L \} \), the prefix closure of \(L\).
  \end{enumerate}

  \end{prop}

By restricting stack operations of a PDA we obtain an important subclass of CF languages called regular languages. 

  \begin{definition}
 A deterministic finite state automaton (FSA) is a PDA M satisfying:
 \begin{enumerate}
 \item[(3')]  \( \Gamma = \varnothing\), 
 \item[(6')] \( D \subset S\times A\), i.e. epsilon transitions are not allowed.  
\end{enumerate}
 \end{definition}

 \begin{definition}
 A language of words \( L \subset A^* \) is called regular if there exists some FSA \(M\) such that \( L = L(M) \). We denote the class of all regular languages by \(\sL_{\rm reg}\).
 
\end{definition}

\begin{exm}\label{ch-ex-1.4}

Consider the alphabet \( A = \{ a , b \} \) and the language of words given by \( L = \{ a^mb^n \mid m,n = 0, 1, \dots \} \).  Figure \ref{fig:FSA1} below depicts a FSA recognizing this language
which demonstrates that it is regular. Note that \(L\) contains the language in Example \ref{ch-ex-1.14}. The FSA reads in a number of \(a\)'s  (possibly none) 
followed by a number of \(b\)'s (possibly none), but cannot keep track of how many letters it has read in since \(\Gamma = \varnothing\). 
\end{exm}

\begin{figure}[H]
\begin{center}
\leavevmode

\includegraphics[trim = 30mm 180mm 20mm 40mm,clip]{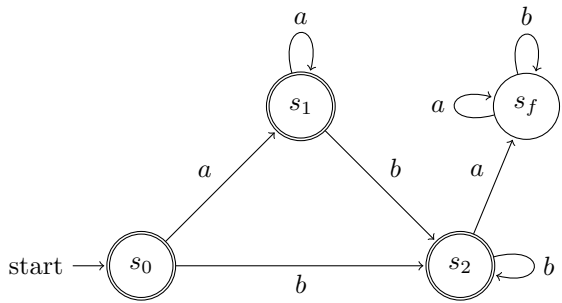}
\caption{A FSA recognizing the language of Example \ref{ch-ex-1.4}. The set of states is \(S=\{s_0,s_1,s_2,s_f\}\) depicted by nodes. The double lines indicate the subset of accept states, and transitions are depicted by labeled arrows. }
  \label{fig:FSA1}
  \end{center}

\end{figure}

Regular languages have nice closure properties (see \cite{E+} or \cite{HU}), some of which do not hold in general for CF languages.
  
  \begin{prop}\label{ch-thm-1.9}
  Let \(L\), \(L_1\), and \(L_2 \in \mathscr{L}_{reg}\)  be regular languages over the alphabet
   \(A\). The following languages are also regular:
  \begin{enumerate}

  \item \(L_1 \cap L_2\),
  \item \(L_1 \cup L_2\), 
  \item \(L_1L_2 = \{ w_1w_2 \mid w_1 \in L_1 \text{ and } w_2 \in L_2\}\),
  \item \(L^*= \{ w_1 \dots w_n \mid n \in \mathbb{N} \text{ and } w_i \in L \}\), the Kleene-star closure of L,
  \item \(L^R= \{w^R=a_n \cdots a_1 \mid w=a_1 \cdots a_n \in L \}\),
  \item \(L_{\rm comp} = A^* \smallsetminus L \),
  \item \(L_{\rm pre} = \{w \mid wu = v \in L \} \), the prefix closure of \(L\).
  \end{enumerate}
  
  Properties 1-5 do not hold in general for CF languages.

\end{prop}

 
 \subsection{Nested words and VPLs}
Visibly pushdown languages (VPLs) are a relatively new class of language introduced by Alur and Madhusudan in \cite{AM}. VPLs are comprised of nested words, which are words with an associated matching relation. One can interpret a VPL as being accepted by a particular type of PDA over an extended alphabet.

  \begin{definition}\label{definition-nestedword}
  A  nested word is a pair \((w,\frown)\) with \(w=a_1\dots a_n \in A^* \) and \(\frown\) is a subset of \(\{-\infty,1,\dotsc,n\} \times \{1, \dotsc,n,\infty\}\) satisfying:
  \begin{enumerate}

 \item (Matching edges go forward) \(i\frown j \Rightarrow i < j\),
 \item (Uniqueness) \(\forall\) \(1 \leq i \leq n\), \(\mid\{j \mid i\frown j\}\mid \leq 1\) and \(\mid\{j \mid j\frown i\}\mid \leq 1\),
 \item (Nesting Property) If \(i_1\frown j_1\) and \(i_2\frown j_2\) with \(i_1 < i_2\), then either \(j_2 < j_1\) or \(j_1 < i_2\).

\end{enumerate}
  \end{definition}
  
  If \(i \frown j\) then we say that \(a_i\) is a call and \(a_j\) is return. When \(i \frown \infty\)  \(a_i\) is called a pending call, and similarly \(-\infty \frown j\) makes \(a_j\) a pending return. If \(a_i\) is neither a call nor a return, then  it is called an internal symbol. 
Let \(NW(A)\) be the set of all nested words over \(A\). 
 
 We may encode nested words over \(A\) by extending the alphabet to the tagged alphabet \(\wt{A} = \cb \negthickspace A \cup A \cup A \! \rb\). The alphabets \(\cb \negthickspace A\) and \(A \! \rb\) are disjoint copies of \(A\) where each element \(a \in A\) is replaced with \( \cb \nms a\) and \(a \nms \rb\), respectively. 
 The idea is that we can tag the letters of any  word \(w \in A^*\) to be either a call or return, or left as an internal symbol. We use $\wt{a} \in \wt{A}$ to denote  an element in the tagged alphabet, so \(\wt{a} \in \{\cb \nms a,a, a \nms \rb \}\).
A tagged word  \(\wt{w} \in \wt{A}^*\) can then be written as  \(\wt{w} = \wt{a}_1 \cdots \wt{a}_n \).
  
  \begin{lem}[2.1 in \cite{AM}] \label{lem-encode}
  There is a natural bijection \(\tau : NW(A) \rightarrow \wt{A}^*\)  given by extending the map: 
   \[
   a_i \mapsto \begin{cases}
         \cb \nms a_i &  \text{if }  a_i \text{ is a call},  \\
         a_i \nms \rb  &  \text{if }  a_i \text{ is a return},  \\
         a_i  &  \text{if }  a_i \text{ is an internal symbol}.  \\
  
         \end{cases}
 \]

  \end{lem}
 
Therefore a language of nested words over $A$ is any subset $L \subset \wt{A}^*$. This leads to a convenient way of defining visibly pushdown languages,  which are also called regular languages of nested words.
 
 \begin{definition}
 Given an alphabet \(A\), a visibly pushdown automaton (VPA) is a PDA over the extended alphabet \(\wt{A}\) satisfying the additional property:
 
 \begin{itemize}

 \item[(7)] \( \delta = \delta_c \cup \delta_i \cup \delta_r \) is the transition function, which depends on the symbol \(\wt{a} \in \{a,\cb \nms a,a \nms \rb\}\) ,
  \[
  \delta(s,\wt{a}) =  \begin{cases}
        \delta_c(s,\cb \nms a) =(s',\gamma_i) & \text{   if     } \wt{a}=\cb \nms a,     \\
       \delta_i(s,a) = s' & \text{   if     } \wt{a}=a, \\
      \delta_r(s,a \nms \rb, \gamma) = s' & \text{   if     } \wt{a}=a \nms \rb.  \\
  
        \end{cases}
 \]
\end{itemize}

\end{definition}

 \begin{definition}
 Given an alphabet \(A\), a language of nested words \(L \subset \wt{A}^*\) is called visibly pushdown (or a regular language of nested words) if there exists a VPA \(\wt{M}\) such that \(L=L(\wt{M})\).
 \end{definition}

The following results emphasize how VPLs are closely related to CF languages.

\begin{thm}[5.1 in \cite{AM}]\label{thm-vpltocfl}
If \(L \subset \wt{A}^*\) is a visibly pushdown language, then it is also a CF language over the alphabet \(\wt{A}\).
\end{thm}

\begin{thm}[5.2 in \cite{AM}]\label{thm-vplontocfl}
If \(L \subset A^*\) is a CF language over the alphabet \(A\) then there exists a visibly pushdown language \(\wt{L} \subset \wt{A}^*\) such that $\rho: \wt{L} \twoheadrightarrow L$ is a surjection, where $\rho$ is the map that forgets the matching relation; i.e. $\rho(\wt{w})=w$.
   
 \end{thm}

This result relies on Lemma \ref{lem-encode} which establishes a correspondence between words over the tagged alphabet \(\wt{A}\) and nested words over \(A\). In other words, for a tagged word $w\in \wt{A}^*$, there is only one way to interpret the tagging such that it satisfies the properties of being a matching relation.
 Lemma \ref{lem-encode} also provides a natural way to define operations on nested words by using standard word operations over the tagged alphabet.  Given an alphabet \(A\), we define concatenation of two words \(\wt{w}_1, \wt{w}_2 \in \wt{A}^*\) to be \(\wt{w_1}\wt{w_2}\), and the prefix of a word \(\wt{w}  \in \wt{A}^*\) of length \(i\) to be \(\wt{w}[i]\). To define reversal consider a word \(\wt{w}=\wt{a}_1 \cdots \wt{a}_n\) with \(\wt{a}_i \in \{\cb \nms a_i,a_i \nms \rb,a_i\}\). The reversal of \(\wt{w}\) is given by \(\wt{w}^R=\wt{b}_{n}\cdots \wt{b}_{1}\), where
  \[ \wt{b}_{i} = \begin{cases}
         a_{i} \nms \rb & \text{ if } \wt{a}_{i}=\cb \nms a_{i}, \\
         \cb \nms a_{i} & \text{ if } \wt{a}_{i}=a_{i} \nms \rb,  \\
         a_{i} & \text{ if } \wt{a}_{i}=a_{i}. \\
  
         \end{cases}
 \]

 Using these operations for nested words, the following analogue of Theorem \ref{ch-thm-1.9} is true for VPLs.
\begin{thm}[3.5--3.7 in \cite{AM}]\label{thm-closurevpl}
  Let \(L\), \(L_1\), and \(L_2\) be VPLs over the extended alphabet
   $\wt{A}$. The following languages are also visibly pushdown:
  \begin{enumerate}

  \item \(L_1 \cap L_2\),
  \item \(L_1 \cup L_2\) ,
  \item \(L_1L_2\),
  \item \(L^*\),
  \item \(L^R=\{\wt{w}^R \mid \wt{w} \in L\}\),
  \item \(L_{\rm comp} = \wt{A}^* \smallsetminus L \),
  \item \(L_{\rm pre} = \{\wt{w} \mid \wt{w}\wt{u} = \wt{v} \in L \} \), the prefix closure of \(L\).
  \end{enumerate}
  
  \end{thm}


\section{Formal Languages and the Word Problem}\label{sec-wp}

Here we show how formal languages are typically used to study the word problem for groups. Given a finitely generated group with presentation \(G = \langle X \mid R \rangle\),  consider the alphabet \(A = X \cup X^{-1}\). There is a canonical monoid epimorphism \(\pi: A^* \twoheadrightarrow G\) taking \(w \in A^*\) to the group element \(\pi(w)=\overline{w}\) that it represents.
For any two words \(w_1, w_2 \in A^*\), the word problem asks for an algorithm to check whether \(\overline{w_1}=\overline{w_2}\) in G. Taking \(w=w_1w_2^{-1}\) this is equivalent to asking for an algorithm that checks whether a given word \(w \in A^*\) is equal to the identity in \(G\).  
 \begin{definition}\label{defn-wp}
Given a formal language class \(\sL\), we say that a finitely generated group \(G=\langle X \mid R\rangle\) admits an \(\sL\) word problem with respect to \(A\) if there exists an \(L\in \sL\) such that 
 \[L=W_A(G) = \{w \in A^* \mid \overline{w}=1\}.\]
 \end{definition}

For regular languages, Anisimov showed that this notion is independent of the choice of generating set. 

\begin{lem}[Anisimov \cite{An}]\label{lem-regwpcog}
Let $\langle X_1 \mid R_1 \rangle$ and $\langle X_2 \mid R_2 \rangle$ be two presentations of a finitely generated group $G$. Then $G$ admits a regular word problem with respect to $A_1$ if and only $G$ admits a regular word problem with respect to $A_2$. 
\end{lem}

This means that having a regular word problem is truly a property of the group itself, as it does not depend on any particular presentation. It is often said the such a property is invariant under a change of generators. This allows for a complete characterization of groups with a regular word problem.   

\begin{thm}[Anisimov \cite{An}]\label{thm-regfinwp} 
A finitely generated group \(G\) has a regular word problem if and only if \(G\) is finite.
\end{thm}

Muller and Schupp proved analogous results for context-free languages in \cite{MS}. For example, Lemma 2 of \cite{MS} shows that, for context-free languages, the notion of a group \(G\) admitting a CF word problem is
independent of the choice of generating set. Before stating their main result,  recall
that a  group \(G\) is called \emph{virtually free} if it contains a free subgroup of finite index.

\begin{exm} 
Any free product \(G_1 \ast G_2\) where \(G_1\) and \(G_2\) are finite is virtually free. This follows from considering the natural map from \(G_1 \ast G_2\) to the direct product \(G_1 \times G_2\) which is finite. Nielsen showed in \cite{Ni} that  we have the following exact sequence
\[1 \rightarrow F_C \rightarrow G_1 \ast G_2 \rightarrow G_1 \times G_2 \rightarrow 1 .\]
The  kernel \(F_C\) is the free group generated by all the commutators \(C= \{ g_1g_2g_1^{-1}g_2^{-1} \mid g_1 \in G_1, g_2 \in G_2\}\), which is the required finite index free group. 
\end{exm}
\begin{exm}
The modular group \(PSL(2,\mathbb{Z})\) is isomorphic to \(\mathbb{Z}_2 \ast \mathbb{Z}_3\), hence virtually free by above. 
\end{exm}
\begin{exm}
Let \(G\) be a finite group and \(\psi : G \rightarrow \Aut(F_n)\) a homomorphism. Then \(F_n \rtimes_\psi G\) is virtually free.
\end{exm}

The following is the main result of \cite{MS} and it completely characterizes those groups \(G\) admitting a context-free word problem.

\begin{thm}[Muller and Schupp \cite{MS}]\label{ch-thm-3.9}
A finitely generated group \(G\) 
has a CF word problem if and only if \(G\) is virtually free. 
\end{thm}

 
  \section{Results} \label{sec-vpl}
 
In this section we  prove new results which allow us to  study the word problem of a finitely generated group using visibly pushdown languages of nested words. We denote a presentation of $G$ by $\langle X \mid R \rangle$  and recall that we take \(A = X \cup X^{-1}\).  Also, recall that the word problem of \(G\) is denoted by \(W_A(G)\). A nested word may be denoted by \((w,\frown)\) with \(w \in A^*\), or by \(\wt{w} \in \wt{A}^*\). Finally recall that \(L \subset \wt{A}^*\) denotes a language of nested words.

 \subsection{The Nested Word Problem}
First it is necessary to give an appropriate definition of the word problem for nested words. Nested words have additional data of the matching relation meaning that the standard definition (i.e. \ref{defn-wp} here) does not make sense; Lemma \ref{lem-encode} says that we must work over the extended alphabet $\wt{A}$ when thinking about nested words.

\begin{definition}\label{def-vplwpdef}
Let $G$ be a finitely generated group  with presentation $G=\langle X \mid R \rangle$. We say that $G$  admits a nested word problem with respect to $A$ if there is a nested word language \(L \subset \wt{A}^*\) such that the forgetful map
is a surjection $$\rho: L \twoheadrightarrow W_A(G).$$
  
\end{definition}

\begin{rmk}\label{rmk-nwpdefn} We note that this definition includes groups with a regular word problem (i.e. finite groups)  in the following sense. The regular language $W_A(G)=L$ for a finite group $G$ can be considered as a VPL by identifying $A$ with the internal symbols of   $\wt{A}$. In this case we have a natural bijection $\rho: L \rightarrow W_A(G)$.  Of course, the full pre-image of $W_A(G)$  under $\rho$ also maps to the word problem.  In other words, the nested word problem for finite groups essentially ignores any matching relations. 
\end{rmk}

We may now state our main theorem which relates the nested word problem to groups with a context-free word problem. 

\begin{thm}\label{thm-nwp}
A finitely generated group \(G\) admits a VPL nested word problem if and only if it is virtually free.
\end{thm}
\begin{proof}
If $G$ admits a nested word problem, then there exists a VPL \(L \subset \wt{A}^*\) such that \(\rho\) is a surjection. Note that \(\rho\) is also a language homomorphism (see \cite{HU}). By Theorem \ref{thm-vpltocfl}, \(L\) is context-free over \(\wt{A}\). It is known  that context-free languages are closed under language homomorphism, hence \(\rho(L)=W_A(G)\) is also context-free, and \(G\) is virtually free by \ref{ch-thm-3.9}. 
Conversely, if $G$ is virtually free there is a CF language $L$ such that $L=W_A(G)$. Theorem \ref{thm-vplontocfl} then gives a VPL $L \subset \wt{A}^*$ mapping onto $W_A(G)$. 
\end{proof}

\begin{cor}\label{cor-changeofgen}
Let $\langle X_1 \mid R_1\rangle$ and $\langle X_2 \mid R_2\rangle$ be two presentations for a finitely generated group $G$.  Then $G$ admits a nested word problem with respect to $A_1$ if and only if $G$ admits a nested word problem with respect to $A_2$.
\end{cor}
\begin{proof}
If $G$ admits a nested word problem with respect to $A_1$, then $G$ is virtually free and $W_{A_1}(G)$ is context-free. Having a context-free word problem is independent of the choice of generating set, hence $W_{A_2}(G)$ is also  context-free. Theorem \ref{thm-vplontocfl} gives a VPL mapping onto $W_{A_2}(G)$ so $G$ admits a nested word problem with respect to $A_2$.
\end{proof}

The advantage of working with nested words is that matching relations provide an intuitive way of thinking about the word problem, as is shown by the next example. 

\begin{exm}
Let $X=\{x_1,x_2,\cdots,x_n\}$ and consider the free group on \(n\) generators, \(F_n=\langle X \rangle\). Any word over \(A\) that represents the identity can be reduced to the empty word by successively deleting pairs of the form \(xx^{-1}\) or \(x^{-1}x\) for \(x \in X\), and this gives rise to a  matching relation defining such words. A particular example is the nested word \(\wt{w}=\cb \nms x_1 (\cb \nms x_3)^{-1} (\cb \nms x_n)^{-1}x_n \nms \rb  x_3 \nms \rb (x_1 \nms \rb)^{-1} \), where the matching relation is given by $\frown \;  = \{ (1,6), (2,5), (3,4)\}$. For the word $x_1x_1^{-1}x_1x_1^{-1}$,  the cancellation could be represented by the matching relation $\frown_1=\{ (1,4), (2,3)\}$ or $\frown_2\{ (1,2), (3,4)\}$. A canonical choice would be $\frown_2$  as it is associated to the path traveled by the word in the Cayley graph of $F_n$.
\end{exm}

This suggests that for the free group we can  find a VPL such that $\rho$ actually provides a bijection with the word problem. In other words, every word representing the identity in the free group is associated to a canonical matching relation. 

\begin{prop} \label{prop-freegpvpl}
Let  $X=\{x_1, \cdots, x_n\}$. The free group on $n$ generators \(F_n=\langle X\rangle\) admits a VPL nested word problem where $\rho$ is a bijection.
\end{prop}
\begin{proof}
We construct a VPA over \(\wt{A}^*\) to recognize the matching relation defined by successive cancellations as described above. The VPA has states \(S= Y \cup \{s_f\)\}, with accept states \( Y =\{(s_0,\wt{a}) \mid \wt{a} \in \wt{A}\} \), initial  state \((s_0,\epsilon)\),  and \(s_f \notin Y\) the fail state.  The stack alphabet is \(\Gamma=A\) and \(\Gamma_y=\varnothing\), implying that accepted words do not contain any pending calls or returns. There are transitions from \((s_0,\wt{a})\) to \(s_f\) described below, but there are no transitions out of \(s_f\). 

The machine keeps track of adjacent trivial relations using the second component of \((s_0,\wt{a})\) as follows. Consider the machine in state \((s_0,\wt{a})\) and reading in the next letter \(\wt{a}'\) such that \(\rho(\wt{a})^{-1}=\rho(\wt{a}')\).  It transitions to \((s_0,\epsilon)\) only if $\wt{a}=\cb \nms a$ and $\wt{a}'= a \nms \rb$; otherwise it transitions to $s_f$.  This ensures the tagging corresponds to the cancellation associated to the path of the word in the Cayley graph. 

In the case where \(\rho(\wt{a})^{-1} \neq \rho(\wt{a}')\) the action of the machine is determined by the tagging of \(\wt{a}'\). If an internal symbol is read the transition is to \(s_f\). On reading a call the underlying letter \(\rho(\wt{a})\) is written to the stack and the machine transitions to \((s_0,\wt{a}')\). Finally, on reading a return, \(\rho(\wt{a}')\) is compared to the letter on the top of the stack. If the pair is of the form \(xx^{-1}\) or \(x^{-1}x\) for \(x \in \{x_1,\cdots x_n\}\) the machine transitions to  \((s_0,\wt{a}')\); if not the machine transitions to \(s_f\).  It also transitions to the fail state on reading a return if the bottom of stack symbol \(\gamma_0\) is exposed. 
\end{proof}

\subsection{Closure Properties and the Nested Word Problem}

Finally, we prove some additional closure properties of VPLs. Though somewhat ad-hoc, these properties - together with the above Proposition \ref{prop-freegpvpl} - enable us to construct further examples of bijections (under $\rho$) between VPLs and the word problem of certain virtually free groups. This establishes a direct correspondence between language theoretic closure properties and the group theoretic properties of direct and semi-direct product. We conjecture that a similar bijection can be constructed in general for any virtually free group, but this lies outside the scope of the current paper.

\begin{definition}
Consider  \(L_1 \subset A^*\) and  \(L_2 \subset B^*\) over the alphabets \(A\) and \(B\) respectively. The shuffle of \(L_1\) and \(L_2\) is denoted by \(L_1 \bowtie L_2\) and is given by
\[ L_1\bowtie L_2=\{ u_1v_1u_2v_2 \cdots u_nv_n \mid u_1\cdots u_n \in L_1, v_1\cdots v_n \in L_2, u_i \in A^*, v_i \in B^* \} .\]
\end{definition}

\begin{lem}\label{thm-closurescramble}
Visibly pushdown languages are closed under shuffle with regular languages; i.e. if \(L \subset \wt{A}^*\) is a VPL and \(L_r \subset B^*\) is regular then \(L \bowtie L_r\) is a VPL. 
\end{lem}

\begin{proof}
 We construct a VPA \(M_{\bowtie}\) that accepts \(L \bowtie L_r\). Take \(\wt{M} = (\wt{A}, S , \Gamma, s_0, \gamma_0,Y,  \Gamma_y,\delta)\) such that \(L=L(\wt{M})\), and \( M_r = (B, S_r , s_0^r , Y_r , \delta_r) \) such that \(L_r=L(M_r)\). We think of \(\wt{M}\) and \(M_r\) as operating side by side.
 The set of states of \(M_{\bowtie}\) is the product \(S \times S_r\), the initial state  \((s_0,s_0^r)\), and the stack alphabet and hierarchical accept states  \(\Gamma\) and \(\Gamma_y\) respectively. The transition function \(\Delta\) for \(M_{\bowtie}\) is then defined based on the symbol being read in. If in state \((s,s_r)\) and reading in \(\wt{a} \in \wt{A}\), the transition is given by \(\Delta((s,s_r),\wt{a})=(\delta(s,\wt{a}),s_r)\), where the appropriate stack operation takes place as it would have in \(\wt{M}\). On reading  \(b \in B\) while in state \((s,s_r)\), the transition is given by \(\Delta((s,s_r),b)=(s,\delta_r(s_r,b))\). The set of accept states is \(Y_{\bowtie}=\{(y,y_r)\mid y \in Y, y_r \in Y_r\}\), and a processed word is accepted if  \(M_{\bowtie}\) is in an accept state with stack contents \(\chi_y \in \Gamma_y^*\).
 \end{proof}

\begin{definition}
Let \(L \subset A\) be a language over \(A\). A finite re-labeling of \(L\) is a map \(\Phi: L \rightarrow A^*\) satisfying:
  \begin{enumerate}

  \item \(|w|=|\Phi(w)|\),
  \item There exists a FSA \(M_{\Phi}\) over \(A \times A\) such that \(L(M_{\Phi}) = \{(w,w')  \mid \Phi(w)=w'\}.\)
  \end{enumerate}
\end{definition}
\begin{rmk}\label{rmk-finiterelabeling1}
To apply this definition to VPLs it is necessary to add to Condition 1 the requirement that \(\Phi((w,\frown))=(w',\frown)\), i.e. that the matching relation is preserved.    
\end{rmk}
\begin{rmk}\label{rmk-finiterelabeling2}
Condition 2 says that \(\Phi\) may change the letters of any word \(w \in L\) but only in a controlled way in the sense that the re-labeling cannot depend on an unbounded amount of information.    
\end{rmk}

\begin{lem}\label{thm-closurerelabeling}
Visibly pushdown languages are closed under finite re-labeling; i.e. if \(L \subset \wt{A}^*\) is a VPL and \(\Phi: L \rightarrow \wt{A}^*\) is a finite re-labeling then \(\Phi(L)\) is a VPL. 
\end{lem}

\begin{proof}
The proof is similar to that of Lemma \ref{thm-closurescramble}. The VPA \(\wt{M}\) accepting \(L\) keeps track of the matching relation, while \(M_{\Phi}\) keeps track of the re-labeling. 
 \end{proof}

Using these properties we provide explicit constructions for VPLs that are in bijection with the word problem for two classes of virtually free groups.

\begin{cor} \label{thm-dpvpl}
Let  \(G\) be a finite group. The direct product \(F_n \times G\)  admits a VPL nested word problem where $\rho$ is a bijection.
\end{cor}

\begin{proof}
Again take $X=\{ x_1, \cdots, x_n\}$ and consider $F_n=\langle X \rangle$ with alphabet \(A = X \cup X^{-1}\). For \(|G|=m\) take the alphabet \(B=G=\{y_1,\cdots,y_m\}\). Denote by \(\mathbb{A}\) the union \(\mathbb{A}=A \cup B\).  By Proposition \ref{prop-freegpvpl} we have a VPL \(L \subset \wt{A}^*\) such that \(\rho: L \rightarrow W_A(F_n)\) is a bijection, and by Theorem \ref{thm-regfinwp} we have a regular language \(L_r=W_B(G)\). The language \(L \bowtie L_r\) is a VPL by Lemma \ref{thm-closurescramble} and \[\rho: L \bowtie L_r \rightarrow W_{\mathbb{A}}(F_n \times G) \] is a bijection.
\end{proof}

\begin{cor} \label{thm-sdpvpl}
Let  \(S_m\) be the symmetric group on \(m \leq n\) letters, and consider the canonical homomorphism  \(\psi: S_m \rightarrow \Aut(F_n)\) where \(\sigma \in S_m\) acts on \(F_n\) by permuting generators. The semi-direct product \(F_n \rtimes_\psi S_m\)  admits a VPL nested word problem where $\rho$ is a bijection.
\end{cor}

\begin{proof}
As above take \(L \subset\wt{A}^*\) for the word problem of \(F_n\) and \(L_r=W_B(S_m)\), with \(\mathbb{A}=A \cup B\). The map \(\psi\) induces a finite re-labeling \(\Phi: L \bowtie L_r \rightarrow \wt{\mathbb{A}}^*\). The language \(\Phi(L \bowtie L_r)\) is a VPL by Lemmas \ref{thm-closurescramble} and \ref{thm-closurerelabeling} and  \[\rho: \Phi(L \bowtie L_r) \rightarrow W_{\mathbb{A}}(F_n \rtimes_\psi S_m)\] is a bijection.
\end{proof}  

\section*{Acknowledgments} 
I am grateful to my advisor Hans U. Boden for guidance, encouragement, and many useful discussions. I would also like to thank Kim P. Huynh, Caroline Junkins, and Kyle Vincent for commenting on drafts of the paper. This work was partially supported by NSERC.

 

\bibliographystyle{plainnat} 
\renewcommand{\bibname}{References} 

\end{document}